\documentclass[12pt,reqno]{amsart}
\usepackage{graphicx}
\usepackage{amscd}
\usepackage{amsmath}
\usepackage{epsfig}
\usepackage{amsfonts}
\usepackage{amssymb}

\catcode`\@=11
\@namedef{subjclassname@2010}{  \textup{2010} Mathematics Subject Classification}
\catcode`\@=12
\def\func#1{\mathop{\rm #1}}
\providecommand{\U}[1]{\protect\rule{.1in}{.1in}}
\providecommand{\U}[1]{\protect\rule{.1in}{.1in}}
\allowdisplaybreaks
\newtheorem{theorem}{Theorem}
\theoremstyle{plain}

\newtheorem{lemma}{Lemma}

\theoremstyle{remark}
\newtheorem*{remark}{Remark}

\numberwithin{equation}{section}

\input{tcilatex}

\begin{document}
\title[Electromagnetic-Field Quantization]{On the Problem of
Electromagnetic-Field Quantization}
\author{Christian Krattenthaler}
\address{Fakult\"{a}t f\"{u}r Mathematik, Universit\"{a}t Wien,
Nordbergstrasze 15, A-1090 Vienna, Austria}
\urladdr{http://www.mat.univie.ac.at/\lower0.5ex\hbox{\~{}}kratt}
\author{Sergey I. Kryuchkov}
\address{School of Mathematical and Statistical Sciences, Arizona State
University, Tempe, AZ 85287--1804, U.S.A.}
\email{sergeykryuchkov@yahoo.com}
\author{Alex Mahalov}
\address{School of Mathematical and Statistical Sciences, Arizona State
University, Tempe, AZ 85287--1804, U.S.A.}
\email{mahalov@asu.edu}
\author{Sergei K. Suslov}
\address{School of Mathematical and Statistical Sciences, Arizona State
University, Tempe, AZ 85287--1804, U.S.A.}
\email{sks@asu.edu}
\urladdr{http://hahn.la.asu.edu/\symbol{126}suslov/index.html}
\date{April 9, 2013}

\begin{abstract}
We consider the radiation field operators in a cavity with varying
dielectric medium in terms of solutions of Heisenberg's equations of motion for the
most general
one-dimensional quadratic Hamiltonian. Explicit solutions of these equations
are obtained and applications to the radiation field quantization, including
randomly varying media, are briefly discussed.
\end{abstract}

\maketitle

\section{Canonical Quantization}

Radiation field quantization
in the vacuum was introduced in original works of Born, Heisenberg and
Jordan \cite{BornJordan1925}, \cite{BornHeisenbergJordan1925} on
matrix mechanics (see also
books on quantum electrodynamics \cite{Akh:Ber}, \cite{Ber:Lif:Pit}, \cite%
{Bia:Bia75} and quantum optics \cite{KlauderSudarshan}, \cite{Louisell73},
\cite{Schleich01}, \cite{Scully:Zubairy97}, \cite{WallsMilburn08}). A modern
mathematical approach to quantization of mechanical systems is discussed in
detail, for example, in \cite{BerezinShubin}, \cite{Faddeyev69}, \cite{Fadd:Yakub}, \cite%
{Takhtajan}, \cite{Teschl}, and/or \cite{WeinbergQM}
(see also \cite{Klauder12} and the references therein).
For a classical Hamiltonian system one
replaces canonically conjugate coordinates and momenta by time-dependent
operators $q_{\lambda }(t)$ and $p_{\lambda }(t)$ that satisfy the
commutation rules%
\begin{equation}
\left[ q_{\lambda }(t),\,q_{\mu }(t)\right] =\left[ p_{\lambda }(t),\,p_{\mu
}(t)\right] =0,\qquad \left[ q_{\lambda }(t),\,p_{\mu }(t)\right] =i\hbar
\delta _{\lambda \mu }.  \label{CommutatorsPQ}
\end{equation}%
The time-evolution is determined by the Heisenberg equations of motion \cite%
{HeisenbergQM}:%
\begin{equation}
\frac{d}{dt}p_{\lambda }(t)=\frac{i}{\hbar }\left[ p_{\lambda }(t),\,%
\mathcal{H}\right] ,\qquad \frac{d}{dt}q_{\lambda }(t)=\frac{i}{\hbar }\left[
q_{\lambda }(t),\,\mathcal{H}\right] ,  \label{HeisenbergEquationsPQ}
\end{equation}%
with appropriate initial conditions.\footnote{%
The standard form of Heisenberg's equations can be obtained by the time
reversal $t\rightarrow -t.$}

Traditionally, the electromagnetic-field quantization is considered under
the assumption that the field occupies an empty box \cite{Akh:Ber}, \cite%
{Ber:Lif:Pit}, \cite{Dod:Klim:Nik93}, \cite{Fey:Hib}. The quantization of
the field in a uniform dielectric medium was discussed in \cite{Glauber91},
\cite{Hillery84}, \cite{HutterBarnet92}, \cite{JauchWatson48}, \cite%
{JauchWatson48a}, \cite{Knolletal87} (see also \cite{Dod:Klim:Nik93}, \cite%
{Drummomd90}, \cite{SuttorpWubs04} and the references therein). Yet the
problem of electromagnetic-field quantization in time-dependent nonuniform
linear nondispersive media remains an active research topic up to now \cite%
{BeiLiu11}, \cite{Bialynicka-Birula87}, \cite{Choi04}, \cite{Choi10}, \cite%
{Dod:Klim:Nik93}, \cite{GrunerWelsch96}, \cite{Hillery09}, \cite{Horsley12},
\cite{Lenac03}, \cite{LobMos91a}, \cite%
{Manetal10}, \cite{Pedrosa11}, \cite{Pedrosa11Conf}, \cite{Philbin10}, \cite%
{Raabel07}, \cite{SuttorpWubs04}, \cite{Vasylyevetal10}.

In the present letter, we study the radiation field operators in a varying
dispersive medium, which is mathematically described by the most general
phenomenological quadratic Hamiltonian $\mathcal{H}$ in an abstract Hilbert
space. (We concentrate on a single photon cavity mode, say $\upsilon ,$ with
frequency $\omega _{\upsilon }=1$ and use the units $c=\hbar =1.$ From now
on, we shall usually omit the indices when dealing with the single mode
under consideration.) In particular, our approach gives a natural
description of squeezed photons that can be created as a result of
parametric amplification of quantum fluctuations in the dynamic Casimir
effect \cite{Dodonov09}, \cite{Dodonov10}, \cite{Lahetal11}, \cite{Man'ko91}%
, \cite{Naylor12}, \cite{Schuetzholdetal98}, \cite{Wilsonetal11} and/or by similar dynamical
amplification mechanisms including the Unruh effect \cite{Unruh76} and
Hawking radiation \cite{BirrelDavies82}, \cite{Hawking74}, \cite{Hawking75}.
It is also useful for quantum fields propagating in nonstationary external
potentials \cite{Bialynicka-Birula87}, \cite{LobMos93}, \cite{Vasylyevetal10}%
, \cite{Schuetzholdetal98}, and for photon quantization in randomly varying media.

\section{Solution of Heisenberg Equations for Nonautonomous Quadratic Systems%
}

Our main result is the following.

\begin{theorem}
The solution of the Heisenberg equations of motion %
\eqref{HeisenbergEquationsPQ} for the nonautonomous quadratic Hamiltonian%
\begin{equation}
H=a(t)p^{2}+b(t)x^{2}+c(t)xp-id(t)-f(t)x-g(t)p  \label{Hamiltonian}
\end{equation}%
($a,$ $b,$ $c,$ $d,$ $f,$ and $g$ are suitable real-valued functions of time
only) has the form%
\begin{equation}
p_{\lambda }(t)=\frac{\widehat{b}(t)-\widehat{b}^{\dagger }(t)}{i\sqrt{2}}%
,\qquad q_{\lambda }(t)=\frac{\widehat{b}(t)+\widehat{b}^{\dagger }(t)}{%
\sqrt{2}}.  \label{pqQED}
\end{equation}%
Here, the time-dependent annihilation $\widehat{b}(t)$ and creation $%
\widehat{b}^{\dagger }(t)$ operators are given by the Ansatz%
\begin{align}
\widehat{b}(t)&=\frac{e^{-2i\gamma \left( t\right) }}{\sqrt{2}}\left( \beta
\left( t\right) x+\varepsilon \left( t\right) +i\frac{p-2\alpha \left(
t\right) x-\delta \left( t\right) }{\beta \left( t\right) }\right) ,  \notag
\\
\widehat{b}^{\dagger }(t)&=\frac{e^{2i\gamma \left( t\right) }}{\sqrt{2}}%
\left( \beta \left( t\right) x+\varepsilon \left( t\right) -i\frac{p-2\alpha
\left( t\right) x-\delta \left( t\right) }{\beta \left( t\right) }\right)
\label{aacross(t)QED}
\end{align}%
in terms of solutions of the Ermakov-type system%
\begin{equation}
\frac{d\alpha }{dt}+b+2c\alpha +4a\alpha ^{2}=a\beta ^{4},  \label{SysA}
\end{equation}%
\begin{equation}
\frac{d\beta }{dt}+\left( c+4a\alpha \right) \beta =0,  \label{SysB}
\end{equation}%
\begin{equation}
\frac{d\gamma }{dt}+a\beta ^{2}=0  \label{SysC}
\end{equation}%
and%
\begin{equation}
\frac{d\delta }{dt}+\left( c+4a\alpha \right) \delta -f-2g\alpha =2a\beta
^{3}\varepsilon ,  \label{SysD}
\end{equation}%
\begin{equation}
\frac{d\varepsilon }{dt}-\left( g-2a\delta \right) \beta =0,  \label{SysE}
\end{equation}%
\begin{equation}
\frac{d\kappa }{dt}-g\delta +a\delta ^{2}=a\beta ^{2}\varepsilon ^{2}.
\label{SysF}
\end{equation}%
The time-independent (self-adjoint) operators $x$ and $p$ obey the canonical
commutation rule $\left[ x,p\right] =i$ in an abstract
(complex) Hilbert space which implies that the relation%
\begin{equation}
\widehat{b}(t)\widehat{b}^{\dagger }(t)-\widehat{b}^{\dagger }(t)\widehat{b}%
(t)=1  \label{commutatora(t)across(t)QED}
\end{equation}%
holds at all times.
\end{theorem}

\begin{proof}
These results can be verified by a direct, but somewhat tedious, calculation
when one expands the solution in generators of the Heisenberg--Weyl algebra,
namely $\left\{ 1,x,p\right\} ,$ with undetermined time-dependent complex
coefficients and simplifies the commutators. The substitutions (\ref{pqQED}%
)--(\ref{aacross(t)QED}) allow us to derive equations (\ref{SysA})--(\ref%
{SysE}), say from the first Heisenberg equation (\ref{HeisenberEquations})
below. (A \textsl{Mathematica} based proof is available on the article's
website; see notebook \textsl{HeisenbergOscillators.nb} and \cite{Kasman}
for an important program ingredient.\footnote{%
See also \textsl{Koutschan.nb} \cite{Kouchan11}.}) Equation~(\ref{SysF}),
which determines the global phase of the corresponding Fock states in the
Schr\"{o}dinger representation, does not show up in this proof, but will
appear later (see Lemma~2).
\end{proof}

By back substitution, we see that the operators $\widehat{b}(t) $ and $%
\widehat{b}^{\dagger }(t) $ are solutions of the Heisenberg equations%
\begin{equation}
\frac{d}{dt}\widehat{b}(t) =i\left[ \ \widehat{b}(t) ,\,H\right] ,\qquad
\frac{d}{dt}\widehat{b}^{\dagger }(t) =i\left[ \ \widehat{b}^{\dagger }(t)
,\,H\right] ,  \label{HeisenberEquations}
\end{equation}%
subject to the initial conditions%
\begin{align}
\widehat{b}(0) &=\frac{e^{-2i\gamma (0) }}{\sqrt{2}}\left( \beta (0)
x+\varepsilon (0) +i\frac{p-2\alpha (0) x-\delta (0) }{\beta (0) }\right) ,
\notag \\
\widehat{b}^{\dagger }(0) &=\frac{e^{2i\gamma (0) }}{\sqrt{2}}\left( \beta
(0) x+\varepsilon (0) -i\frac{p-2\alpha (0) x-\delta (0) }{\beta ( 0) }%
\right) .  \label{HeisenbergInitialData}
\end{align}

To a certain extent, the creation and annihilation operators (\ref%
{aacross(t)QED}) allow us to incorporate the Schr\"{o}dinger symmetry group
of the harmonic oscillator, originally found
in coordinate representation \cite{Niederer72}, \cite{Niederer73}, into a
more abstract Heisenberg picture. (For the sake of simplicity, we have
restricted ourselves to a single photon mode $\upsilon $ with frequency $%
\omega _{\upsilon }=1;$ see \cite{KrySusVegaMinimum} for a detailed
investigation of the special case of uniform media.)

A concept of dynamical invariants for generalized harmonic oscillators,
which is crucial for constructing the corresponding Fock states from our
creation and annihilation operators, has been recently revisited in \cite%
{Cor-Sot:Sua:SusInv}, \cite{SanSusVin}, and \cite{Suslov10} (see \cite%
{Dod:Mal:Man75}, \cite{Dodonov:Man'koFIAN87}, \cite{Lewis:Riesen69}, \cite%
{Malkin:Man'ko79}, \cite{Malk:Man:Trif73}, \cite{Schradeetal995} and the references therein for
classical accounts).

\section{Solving The Ermakov-Type System}

A general solution of (\ref{SysA})--(\ref{SysF}) is given by Lemma~3 of \cite%
{Lan:Lop:Sus} in a real form (see also \cite{Kouchan11} and \cite%
{Lop:Sus:VegaGroup}). In order to proceed to a more compact form, one needs
to recall some notation. The substitution%
\begin{equation}
\alpha =\frac{1}{4a}\frac{\mu ^{\prime }}{\mu }-\frac{d}{2a}  \label{Alpha}
\end{equation}%
reduces the inhomogeneous equation (\ref{SysA}) to the second order ordinary
differential equation%
\begin{equation}
\mu ^{\prime \prime }-\tau (t)\mu ^{\prime }+4\sigma (t)\mu
=c_{0}(2a)^{2}\beta ^{4}\mu ,  \label{CharEq}
\end{equation}%
which has the familiar time-varying coefficients%
\begin{equation}
\tau (t)=\frac{a^{\prime }}{a}-2c+4d,\qquad \sigma (t)=ab-cd+d^{2}+\frac{d}{2%
}\left( \frac{a^{\prime }}{a}-\frac{d^{\prime }}{d}\right) .
\label{TauSigma}
\end{equation}%
(In \eqref{CharEq} and in the rest of the paper, we use a formal `binary'
parameter $c_{0}=0,1$ for the sake of convenience.)

The time-dependent coefficients $\alpha _{0},$ $\beta _{0},$ $\gamma _{0},$ $%
\delta _{0},$ $\varepsilon _{0},$ $\kappa _{0},$ which satisfy the
homogeneous (Riccati-type) system (\ref{SysA})--(\ref{SysF}), are given by
(cf.\ \cite{Cor-Sot:Lop:Sua:Sus}, \cite{Lan:Lop:Sus}, \cite{Suslov10})%
\begin{align}
\alpha _{0}(t)& =\frac{1}{4a(t)}\frac{\mu _{0}^{\prime }(t)}{\mu _{0}(t)}-%
\frac{d(t)}{2a(t)},  \label{A0} \\
\beta _{0}(t)& =-\frac{\lambda (t)}{\mu _{0}(t)},\qquad \lambda (t)=\exp
\left( -\int_{0}^{t}\left( c(s)-2d(s)\right) \,ds\right) ,  \label{B0} \\
\gamma _{0}(t)& =\frac{1}{2\mu _{1}(0)}\frac{\mu _{1}(t)}{\mu _{0}(t)}+\frac{%
d(0)}{2a(0)},  \label{C0}
\end{align}%
and%
\begin{align}
\delta _{0}(t)& =\frac{\lambda (t)}{\mu _{0}(t)}\int_{0}^{t}\left[ \left(
f(s)-\frac{d(s)}{a(s)}g(s)\right) \mu _{0}(s)+\frac{g(s)}{2a(s)}\mu
_{0}^{\prime }(s)\right] \frac{ds}{\lambda (s)},  \label{D0} \\
\varepsilon _{0}(t)& =-\frac{2a(t)\lambda (t)}{\mu _{0}^{\prime }(t)}\delta
_{0}(t)+8\int_{0}^{t}\frac{a(s)\sigma (s)\lambda (s)}{\left( \mu
_{0}^{\prime }(s)\right) ^{2}}\left( \mu _{0}(s)\delta _{0}(s)\right) \,ds
\notag \\
& \kern2cm+2\int_{0}^{t}\frac{a(s)\lambda (s)}{\mu _{0}^{\prime }(s)}\left(
f(s)-\frac{d(s)}{a(s)}g(s)\right) \,ds,  \label{E0} \\
\kappa _{0}(t)& =\frac{a(t)\mu _{0}(t)}{\mu _{0}^{\prime }(t)}\delta
_{0}^{2}(t)-4\int_{0}^{t}\frac{a(s)\sigma (s)}{\left( \mu _{0}^{\prime
}(s)\right) ^{2}}\left( \mu _{0}(s)\delta _{0}(s)\right) ^{2}\,ds  \notag \\
& \kern2cm-2\int_{0}^{t}\frac{a(s)}{\mu _{0}^{\prime }(s)}\left( \mu
_{0}(s)\delta _{0}(s)\right) \left( f(s)-\frac{d(s)}{a(s)}g(s)\right) \,ds
\label{F0}
\end{align}%
($\delta _{0}(0)=-\varepsilon _{0}(0)=g(0)/\left( 2a(0)\right) $ and $\kappa
_{0}(0)=0$), provided that $\mu _{0}$ and $\mu _{1}$ are the standard
(real-valued) solutions of equation (\ref{CharEq}) when $c_{0}=0$
corresponding to the initial conditions $\mu _{0}(0)=0,$ $\mu _{0}^{\prime
}(0)=2a(0)\neq 0$ and $\mu _{1}(0)\neq 0,$ $\mu _{1}^{\prime }(0)=0.$
(Proofs of these facts are outlined in \cite{Cor-Sot:Lop:Sua:Sus} and \cite%
{Cor-Sot:SusDPO}. The integrals are treated in the
most general way, which may include stochastic calculus; see, for example,
\cite{Oksendal00}.)

Here, we would like to present a new compact form of these solutions. Let us
introduce the complex-valued function%
\begin{equation}
z(t) =\left( 2\alpha (0)+\frac{d(0)}{a(0)}\right) \mu _{0}(t)+\frac{\mu
_{1}(t)}{\mu _{1}(0)}+i\beta ^{2}(0)\mu _{0}(t)  \label{ComplexZ}
\end{equation}%
(a complex parametrization of Green's function, linear invariants,
and Wigner functions of
generalized harmonic oscillators are also discussed in \cite{Dod:Man79},
\cite{Dodonov:Man'koFIAN87}, \cite{Har:Ben-Ar:Mann11}, and \cite{Schradeetal995}). Then%
\begin{equation}
z(t) =c_{1}E(t)+c_{2}E^{\ast }(t),  \label{ComplexZE}
\end{equation}%
where the complex-valued solutions are given by%
\begin{equation}
E(t)=\frac{\mu _{1}(t)}{\mu _{1}(0)}+i\mu _{0}(t),\qquad E^{\ast }(t)=\frac{%
\mu _{1}(t)}{\mu _{1}(0)}-i\mu _{0}(t),  \label{ComplexE}
\end{equation}%
and the corresponding complex-valued parameters are defined by
\begin{equation}
c_{1}=\frac{1+\beta ^{2}(0)}{2}-i\left( \alpha (0)+\frac{d(0)}{2a(0)}\right)
,\qquad c_{2}=\frac{1-\beta ^{2}(0)}{2}+i\left( \alpha (0)+\frac{d(0)}{2a(0)}%
\right) ,  \label{ComplexC12}
\end{equation}%
with%
\begin{equation}
c_{1}+c_{2}=1,\qquad \left\vert c_{1}\right\vert ^{2}-\left\vert
c_{2}\right\vert ^{2}=c_{1}-c_{2}^{\ast }=\beta ^{2}(0).
\label{ComplexC12Relations}
\end{equation}%
In addition,%
\begin{equation}
z(0)=c_{1}+c_{2}=1,\qquad z^{\prime }(0)=2ia(0)\left( c_{1}-c_{2}\right) .
\label{ComplInitData}
\end{equation}%
When written in terms of the complex
function $z$ in \eqref{ComplexZ}, the complex conjugate functions $E$ and $%
E^{\ast }$ defined in \eqref{ComplexE} become%
\begin{equation}
E=\frac{c_{1}^{\ast }z-c_{2}z^{\ast }}{\left\vert c_{1}\right\vert
^{2}-\left\vert c_{2}\right\vert ^{2}},\qquad E^{\ast }=\frac{c_{1}z^{\ast
}-c_{2}^{\ast }z}{\left\vert c_{1}\right\vert ^{2}-\left\vert
c_{2}\right\vert ^{2}}  \label{EZ}
\end{equation}%
and%
\begin{equation}
\mu _{0}=\frac{z-z^{\ast }}{2i\left( c_{1}-c_{2}^{\ast }\right) },\qquad
\frac{\mu _{1}}{\mu _{1}(0)}=\frac{\left( c_{1}^{\ast }-c_{2}^{\ast }\right)
z+\left( c_{1}-c_{2}\right) z^{\ast }}{2\left( c_{1}-c_{2}^{\ast }\right) }.
\label{MZ}
\end{equation}

One can readily verify that%
\begin{gather}
\alpha _{0}=\frac{1}{4a}\frac{\left( z-z^{\ast }\right) ^{\prime }}{%
z-z^{\ast }}+\frac{d}{2a},\qquad \beta _{0}=-2i\lambda \frac{%
c_{1}-c_{2}^{\ast }}{z-z^{\ast }},  \notag \\
\gamma _{0}=\frac{\left( c_{1}^{\ast }-c_{2}^{\ast }\right) z+\left(
c_{1}-c_{2}\right) z^{\ast }}{2i\left( z-z^{\ast }\right) }+\frac{d(0)}{2a(0)%
},  \label{ABC0}
\end{gather}%
and equations (\ref{D0})--(\ref{F0}) can also be rewritten in terms of the
function $z$ in view of (\ref{MZ}). Finally, we introduce a second complex
function,
\begin{equation}
\zeta (t)=\varepsilon (0)\beta (0)+i\left( \delta (0)+\varepsilon _{0}\left(
t\right) \right) =c_{3}+i\varepsilon _{0},\qquad c_{3}=\varepsilon (0)\beta
(0)+i\delta (0),  \label{ComplexZeta}
\end{equation}%
and indicate the inverse relations between the essential, real and complex,
parameters:%
\begin{equation}
\alpha (0)=\frac{c_{1}^{\ast }-c_{1}}{2i}-\frac{d(0)}{2a(0)},\qquad \beta
^{2}(0)=c_{1}-c_{2}^{\ast }=\left\vert c_{1}\right\vert ^{2}-\left\vert
c_{2}\right\vert ^{2},  \label{ZAB}
\end{equation}%
and%
\begin{equation}
\delta (0)=\frac{c_{3}-c_{3}^{\ast }}{2i},\qquad \varepsilon (0)=\pm \frac{%
c_{3}+c_{3}^{\ast }}{2\sqrt{\left\vert c_{1}\right\vert ^{2}-\left\vert
c_{2}\right\vert ^{2}}}.  \label{ZDE}
\end{equation}%
Then the solution of the initial value problem for the Ermakov-type system
can be expressed in terms of the complex
function $z$ in \eqref{ComplexZ} as given in the lemma below.

\begin{lemma}
The system \eqref{SysA}--\eqref{SysF} is solved by
\begin{gather}
\alpha =\alpha _{0}+\lambda ^{2}\frac{c_{1}-c_{2}^{\ast }}{2i\left\vert
z\right\vert ^{2}}\ \frac{z+z^{\ast }}{z-z^{\ast }},\quad \beta =\pm \lambda
\frac{\sqrt{\left\vert c_{1}\right\vert ^{2}-\left\vert c_{2}\right\vert ^{2}%
}}{\left\vert z\right\vert },\quad \gamma =\gamma (0) -\frac{1}{2}\arg z,
\label{ABC} \\
\delta =\delta _{0}+\lambda \frac{\zeta z-\zeta ^{\ast }z^{\ast }}{%
2i\left\vert z\right\vert ^{2}},\qquad \varepsilon =\pm \frac{\zeta z+\zeta
^{\ast }z^{\ast }}{2\left\vert z\right\vert \sqrt{\left\vert
c_{1}\right\vert ^{2}-\left\vert c_{2}\right\vert ^{2}}},  \label{CE} \\
\kappa =\kappa (0) +\kappa _{0}+\frac{\left( \zeta ^{2}z+\left. \zeta ^{\ast
}\right. ^{2}z^{\ast }\right) \left( z-z^{\ast }\right) }{8i\left(
c_{1}-c_{2}^{\ast }\right) \left\vert z\right\vert ^{2}},  \label{F}
\end{gather}
with $z$ and $\zeta$ as in \eqref{ComplexZ} and \eqref{ComplexZeta},
respectively.
\textrm{(}The solution of the Ermakov-type equation \eqref{CharEq} is given by $%
\mu =\mu (0) \,\vert z\vert .$\textrm{)}
\end{lemma}

\begin{proof}
This amounts to a straightforward calculation using
Lemma~3 in \cite{Lan:Lop:Sus}.
\end{proof}

As a consequence, one gets%
\begin{gather}
\frac{2i\left( \alpha -\alpha _{0}\right) }{\beta ^{2}}=\frac{z+z^{\ast }}{%
z-z^{\ast }},  \label{CA1} \\
i\left( \alpha -\alpha _{0}\right) +\frac{\beta ^{2}}{2}=\beta ^{2}\frac{z}{%
z-z^{\ast }},\qquad i\left( \alpha -\alpha _{0}\right) -\frac{\beta ^{2}}{2}%
=\beta ^{2}\frac{z^{\ast }}{z-z^{\ast }}  \label{CA2} \\
\varepsilon +i\frac{\delta -\delta _{0}}{\beta }=\frac{\zeta z}{\beta (0)
\left\vert z\right\vert },\qquad \varepsilon -i\frac{\delta -\delta _{0}}{%
\beta }=\frac{\zeta ^{\ast }z^{\ast }}{\beta (0) \left\vert z\right\vert },
\label{CE1} \\
\qquad \varepsilon ^{2}+\left( \frac{\delta -\delta _{0}}{\beta }\right)
^{2}=\varepsilon ^{2}(0) +\left( \frac{\delta (0) +\varepsilon _{0}}{\beta
(0) }\right) ^{2} ,  \label{CE2}
\end{gather}%
and%
\begin{equation}
\kappa =\kappa (0) +\kappa _{0}+\frac{\delta -\delta _{0}}{2\beta }%
\varepsilon -\frac{\varepsilon _{0}+\delta (0) }{2\beta (0) }\varepsilon (0)
.  \label{CKappa}
\end{equation}%
These \textquotedblleft quasi-invariants\textquotedblright\ can be useful,
for example, when making a comparison of calculations done by different
approximation methods.

Examples of explicitly integrable quadratic systems are discussed in \cite%
{Cor-Sot:Sua:SusInv}, \cite{Cor-Sot:SusDPO}, \cite{Dod:Man79}, \cite%
{Dodonov:Man'koFIAN87}, \cite{Lewis:Riesen69}, \cite{Lop:Sus}, \cite%
{Lop:Sus:VegaGroup}, \cite{LopSusVegaHarm}, \cite{Malkin:Man'ko79}, \cite{Schradeetal995},
and \cite{Yuen76} (see also the references therein).

\section{Single Mode Fock States for Nonautonomous Quadratic Hamiltonians}

The time-dependent quadratic operator (see \cite{SanSusVin})%
\begin{equation}
\widehat{E}(t)=\frac{1}{2}\left[ \frac{\left( p-2\alpha x-\delta \right) ^{2}%
}{\beta ^{2}}+\left( \beta x+\varepsilon \right) ^{2}\right] =\frac{1}{2}%
\left[ \widehat{a}(t)\widehat{a}^{\dagger }(t)+\widehat{a}^{\dagger }(t)%
\widehat{a}(t)\right] ,  \label{QIQED}
\end{equation}%
with the defining property%
\begin{equation}
i\frac{d\widehat{E}}{dt}+\widehat{E}H-H\widehat{E}=0,  \label{QDQED}
\end{equation}%
extends the standard Hamiltonian
(and/or number operator) for any given solution of the Ermakov-type system.
(The initial data play a role of integrals/constants of motion
and/or quantum numbers for the creation and annihilation operators in the
Heisenberg representation.) We use the substitution $\widehat{b}%
(t)=e^{-2i\gamma }\widehat{a}(t)$ and $\widehat{b}^{\dagger }(t)=\widehat{a}%
^{\dagger }(t)e^{2i\gamma }$ with $\left[ \widehat{a}(t),\widehat{a}%
^{\dagger }(t)\right] =1.$ The oscillator-type spectrum,%
\begin{equation}
\widehat{E}(t)\left\vert \Psi _{n}(t)\right\rangle =\left( n+\frac{1}{2}%
\right) \left\vert \Psi _{n}(t)\right\rangle ,  \label{EigenValueProblemQED}
\end{equation}%
can be obtained in a standard fashion (with the aid of modified variable
creation and annihilation operators; cf.\ \cite{Akh:Ber}, \cite%
{BerezinShubin}, \cite{Fock32-2}, \cite{Fock34-3}):%
\begin{equation}
\widehat{a}(t)\left\vert \Psi _{n}(t)\right\rangle =\sqrt{n}\ \left\vert
\Psi _{n-1}(t)\right\rangle ,\quad \widehat{a}^{\dagger }(t)\left\vert \Psi
_{n}(t)\right\rangle =\sqrt{n+1}\ \left\vert \Psi _{n+1}(t)\right\rangle .
\label{annandcratoperactionsQED}
\end{equation}%
Here and in what follows, it is convenient to use the orthogonality relation
$\left\langle \Psi _{m}(t),\Psi _{n}(t)\right\rangle =\delta _{mn}\, \lambda
^{-1}$ with $\beta (0)\mu (0)=1.$

Now we can analyze abstract Fock states in the Schr\"{o}dinger
representation (see also \cite{Fock28-2}).

\begin{lemma}
Let $\widehat{a}(t) \left\vert \Psi _{0}(t) \right\rangle =0.$ The dynamic
Fock states given by%
\begin{equation}
\left\vert \psi _{n}(t) \right\rangle =e^{i( 2n+1) \gamma +i\kappa
}\left\vert \Psi _{n}(t) \right\rangle =\frac{e^{i( 2n+1) \gamma +i\kappa }}{%
\sqrt{n!}}\left( \widehat{a}^{\dagger }(t) \right) ^{n}\left\vert \Psi
_{0}(t) \right\rangle  \label{FockState(t)}
\end{equation}%
satisfy the time-dependent Schr\"{o}dinger equation%
\begin{equation}
i\frac{d}{dt}\left\vert \psi _{n}(t) \right\rangle =H\left\vert \psi _{n}(t)
\right\rangle  \label{SchroedingerEqyation(t)}
\end{equation}%
with the general quadratic Hamiltonian \eqref{Hamiltonian} provided that
equation \eqref{SysF} for the global phase holds and $\left\langle \Phi
_{n},d\Phi _{n}/dt\right\rangle =0$ for $\Phi _{n}=\lambda ^{1/2}e^{-i\left(
\alpha x^{2}+\delta x\right) }\Psi _{n}.$
\end{lemma}

\begin{proof}
From (\ref{QDQED})--(\ref{EigenValueProblemQED}), one gets, formally,%
\begin{equation}
\widehat{E}\left( i\frac{d}{dt}\left\vert \psi _{n}(t)\right\rangle
-H\left\vert \psi _{n}(t)\right\rangle \right) =\left( n+\frac{1}{2}\right)
\left( i\frac{d}{dt}\left\vert \psi _{n}(t)\right\rangle -H\left\vert \psi
_{n}(t)\right\rangle \right) .  \label{InvariantEVP}
\end{equation}%
Therefore%
\begin{equation}
i\frac{d}{dt}\left\vert \psi _{n}(t)\right\rangle -H\left\vert \psi
_{n}(t)\right\rangle =c_{n}\left\vert \psi _{n}(t)\right\rangle ,
\label{InvariantCondition}
\end{equation}%
where $\mathop{\rm Im}c_{n}=0$ in view of the normalization condition; see
\cite{Cor-Sot:Sua:SusInv}, \cite{Fock28-2}, \cite{Lewis:Riesen69}. (We assume also that the
vacuum state is nondegenerate.)

Here,%
\begin{equation}
H=ap^{2}+bx^{2}+\frac{c}{2}\left( px+xp\right) +\frac{i}{2}\left(
c-2d\right) -fx-gp,  \label{Hpx}
\end{equation}%
and the position and linear momentum operators are given by%
\begin{align}
x& =\frac{1}{\beta }\left[ \frac{1}{\sqrt{2}}\left( \widehat{a}+\widehat{a}%
^{\dagger }\right) -\varepsilon \right] ,  \label{x2a} \\
p& =\frac{\beta }{i\sqrt{2}}\left( \widehat{a}-\widehat{a}^{\dagger }\right)
+\frac{\sqrt{2}\alpha }{\beta }\left( \widehat{a}+\widehat{a}^{\dagger
}\right) +\delta -\frac{2\alpha \varepsilon }{\beta }.  \label{p2a}
\end{align}%
In terms of creation and annihilation operators, the Hamiltonian takes the
form%
\begin{align}
H& =\left[ \frac{a}{2}\left( \frac{4\alpha ^{2}}{\beta ^{2}}-\beta
^{2}\right) +\frac{b+2c\alpha }{2\beta ^{2}}-\frac{i}{2}\left( c+4a\alpha
\right) \right] \left( \widehat{a}\right) ^{2}  \notag \\
& \kern2cm+\left[ \frac{a}{2}\left( \frac{4\alpha ^{2}}{\beta ^{2}}-\beta
^{2}\right) +\frac{b+2c\alpha }{2\beta ^{2}}+\frac{i}{2}\left( c+4a\alpha
\right) \right] \left( \widehat{a}^{\dagger }\right) ^{2}  \notag \\
& \kern2cm+\frac{1}{2}\left[ a\left( \beta ^{2}+\frac{4\alpha ^{2}}{\beta
^{2}}\right) +\frac{b+2c\alpha }{\beta ^{2}}\right] \left( \widehat{a}%
\widehat{a}^{\dagger }+\widehat{a}^{\dagger }\widehat{a}\right) +\frac{i}{2}%
\left( c-2d\right)  \notag \\
& \kern2cm+\sqrt{2}\left[ \frac{4a\alpha +c}{2\beta }\left( \delta -\frac{%
2\alpha \varepsilon }{\beta }\right) -\frac{\varepsilon }{\beta ^{2}}\left(
b+c\alpha \right) -\frac{f+2g\alpha }{2\beta }\right.  \notag \\
& \kern3.5cm+\left. \frac{i}{2}\left( \beta \left( g-2a\delta \right)
+\varepsilon \left( c+4a\alpha \right) \right) \right] \widehat{a}  \notag \\
& \kern2cm+\sqrt{2}\left[ \frac{4a\alpha +c}{2\beta }\left( \delta -\frac{%
2\alpha \varepsilon }{\beta }\right) -\frac{\varepsilon }{\beta ^{2}}\left(
b+c\alpha \right) -\frac{f+2g\alpha }{2\beta }\right.  \notag \\
& \kern3.5cm-\left. \frac{i}{2}\left( \beta \left( g-2a\delta \right)
+\varepsilon \left( c+4a\alpha \right) \right) \right] \widehat{a}^{\dagger }
\notag \\
& \kern2cm+a\left( \delta -\frac{2\alpha \varepsilon }{\beta }\right) ^{2}+%
\frac{\varepsilon }{\beta }\left( f+\frac{b\varepsilon }{\beta }\right)
-\left( \delta -\frac{2\alpha \varepsilon }{\beta }\right) \left( g+\frac{%
c\varepsilon }{\beta }\right) ,  \label{H2a1a}
\end{align}%
where we have corrected a typo in \cite{SanSusVin}. As a result, by (\ref%
{annandcratoperactionsQED})--(\ref{EigenValueProblemQED}) we have
\begin{multline}
\lambda \mathop{\rm Re}\left\langle \Psi _{n},H\Psi _{n}\right\rangle
=\left( n+\frac{1}{2}\right) \left[ a\left( \beta ^{2}+\frac{4\alpha ^{2}}{%
\beta ^{2}}\right) +\frac{b+2c\alpha }{\beta ^{2}}\right]  \label{Psi2A} \\
+a\left( \delta -\frac{2\alpha \varepsilon }{\beta }\right) ^{2}+\frac{%
\varepsilon }{\beta }\left( f+\frac{b\varepsilon }{\beta }\right) -\left(
\delta -\frac{2\alpha \varepsilon }{\beta }\right) \left( g+\frac{%
c\varepsilon }{\beta }\right)
\end{multline}%
in terms of solutions of the Ermakov-type system.

In order to complete the proof, one can repeat the evaluation of Berry's
phase \cite{Berry84} for generalized harmonic oscillators, given in \cite%
{SanSusVin} in coordinate representation, in a more abstract form. Indeed,%
\begin{equation}
\left\langle \psi _{n},H\psi _{n}\right\rangle =\left\langle \psi _{n},i%
\frac{d}{dt}\psi _{n}\right\rangle =-\lambda ^{-1}\left[ (2n+1)\frac{d\gamma
}{dt}+\frac{d\kappa }{dt}\right] +\left\langle \Psi _{n},i\frac{d}{dt}\Psi
_{n}\right\rangle =\left\langle \Psi _{n},H\Psi _{n}\right\rangle .
\label{TransPsy}
\end{equation}%
In general,%
\begin{equation}
\Psi _{n}=e^{(1/2)\int \left( c-2d\right) \ dt}e^{i\left( \alpha
x^{2}+\delta x\right) }\Phi _{n},\qquad \left\langle \Phi _{m},\Phi
_{n}\right\rangle =\delta _{mn},  \label{PHY}
\end{equation}%
and%
\begin{align}
\lambda \left\langle \Psi _{n},\frac{d\Psi _{n}}{dt}\right\rangle &
=i\left\langle \Phi _{n},\left( \frac{d\alpha }{dt}x^{2}+\frac{d\delta }{dt}%
x\right) \Phi _{n}\right\rangle +\frac{1}{2}\left( c-2d\right) +\left\langle
\Phi _{n},\frac{d\Phi _{n}}{dt}\right\rangle  \notag \\
& =i\frac{d\alpha }{dt}\lambda \left\langle \Psi _{n},x^{2}\Psi
_{n}\right\rangle +i\frac{d\delta }{dt}\lambda \left\langle \Psi _{n},x\Psi
_{n}\right\rangle +\frac{1}{2}\left( c-2d\right) ,  \label{LHS}
\end{align}%
where $\left\langle \Phi _{n},d\Phi _{n}/dt\right\rangle =0$ by our
hypothesis. Here,%
\begin{equation}
\lambda \left\langle \Psi _{n},x\Psi _{n}\right\rangle =-\varepsilon \beta
^{-1},\qquad \lambda \left\langle \Psi _{n},x^{2}\Psi _{n}\right\rangle
=\beta ^{-2}\left( \varepsilon ^{2}+n+\frac{1}{2}\right) .
\label{MartixElementsXX2}
\end{equation}%
As a result,%
\begin{equation}
\lambda \mathop{\rm Re}\left( i\left\langle \Psi _{n},\frac{d\Psi _{n}}{dt}%
\right\rangle \right) =-\beta ^{-2}\left( \varepsilon ^{2}+n+\frac{1}{2}%
\right) \frac{d\alpha }{dt}+\varepsilon \beta ^{-1}\frac{d\delta }{dt}.
\label{LHSErmakov}
\end{equation}%
Finally, in view of (\ref{InvariantCondition}), (\ref{Psi2A}) and (\ref%
{LHSErmakov}), we obtain (after some simplification) that
\begin{equation*}
c_{n}=\left\langle \Phi _{n},id\Phi _{n}/dt\right\rangle =0
\end{equation*}%
for any given solution of the Ermakov-type system (\ref{SysA})--(\ref{SysF}%
). This completes the proof.
\end{proof}

\begin{remark}
In coordinate representation, when $\Phi _{n}$ are, essentially, the
real-valued stationary orthonormal wave functions for the simple harmonic
oscillator with respect to the new variable $\xi =\beta x+\varepsilon $ (see
\cite{La:Lif}, \cite{Lan:Lop:Sus}, \cite{Ni:Su:Uv}, and \cite{SanSusVin} for
more details), the equation $\left\langle \Phi _{n},d\Phi
_{n}/dt\right\rangle =0$ is valid due to the normalization condition $%
\left\Vert \Phi _{n}\right\Vert ^{2}=1.$ In general, one gets $%
c_{n}=\left\langle \Phi _{n},id\Phi _{n}/dt\right\rangle $, and the previous
connection is associated with a transport law for line bundles in the
Hilbert space, namely, the change $d\Phi _{n}$ is orthogonal to $\Phi _{n}$
\cite{Simon83}.
\end{remark}

The last
lemma can be reformulated in terms of an analog of Berry's phase \cite%
{Berry84},%
\begin{equation}
\frac{d\theta _{n}}{dt}=\lambda \mathop{\rm Re}\left\langle \psi _{n},i\frac{%
d}{dt}\psi _{n}\right\rangle -\frac{d\varphi _{n}}{dt},  \label{ThetaPhiH}
\end{equation}%
where $\varphi _{n}(t)=-(2n+1)\gamma (t)$ is the dynamical phase \cite%
{SanSusVin} (see also \cite{Suslov12} for an example).

\begin{lemma}
The Fock states, given by \eqref{FockState(t)} in terms of solutions of the
Ermakov-type system, satisfy the Schr\"{o}dinger equation %
\eqref{SchroedingerEqyation(t)} if and only if
the two expressions {\rm{(40)}} and {\rm{(48)}} in \cite{SanSusVin}
for the derivative of the Berry phase become equivalent.
\end{lemma}

\begin{remark}
In \cite{SanSusVin} it is shown that the two expressions
(40) and (48) are equivalent in coordinate representation,
when a particular form of coordinate and momentum are taken
together with the $L^2$ inner product. The point of the above
lemma is that this happens in an abstract Hilbert space
if and only if the Fock states
satisfy the Schr\"odinger equation~\eqref{SchroedingerEqyation(t)}.
\end{remark}

\section{Expectation Values and Variances}

By (\ref{x2a})--(\ref{p2a}) and (\ref{annandcratoperactionsQED}), one gets%
\begin{eqnarray}
&&\left\langle \psi _{n}(t),x\psi _{n}(t)\right\rangle =-\frac{\lambda
\varepsilon }{\beta },\qquad \qquad \quad \overline{x}=\frac{\left\langle
x\right\rangle }{\left\langle 1\right\rangle }=\frac{\left\langle \psi
_{n},x\psi _{n}\right\rangle }{\left\langle \psi _{n},\psi _{n}\right\rangle
}=-\frac{\varepsilon ( t) }{\beta ( t) },
\label{x2aQED} \\
&&\left\langle \psi _{n}(t),p\psi _{n}(t)\right\rangle =\lambda \left(
\delta -\frac{2\alpha \varepsilon }{\beta }\right) ,\quad \overline{p}=\frac{%
\left\langle p\right\rangle }{\left\langle 1\right\rangle }=\frac{%
\left\langle \psi _{n},p\psi _{n}\right\rangle }{\left\langle \psi _{n},\psi
_{n}\right\rangle }=\delta ( t) -\frac{2\alpha ( t)
\varepsilon ( t) }{\beta ( t) }.  \label{p2aQED}
\end{eqnarray}%
(The relation to the Ehrenfest theorem is discussed in \cite{Lan:Lop:Sus}.)

The standard deviations,%
\begin{equation}
\left( \delta p\right) ^{2}=\frac{\left\langle \left( \Delta p\right)
^{2}\right\rangle }{\left\langle 1\right\rangle }=\overline{\left(
p^{2}\right) }-\left( \overline{p}\right) ^{2},\qquad \left( \delta x\right)
^{2}=\frac{\left\langle \left( \Delta x\right) ^{2}\right\rangle }{%
\left\langle 1\right\rangle }=\overline{\left( x^{2}\right) }-\left(
\overline{x}\right) ^{2},  \label{bdp}
\end{equation}%
are given by%
\begin{equation}
\left( \delta p\right) ^{2}=\left( n+\frac{1}{2}\right) \left( \beta ^{2}+%
\frac{4\alpha ^{2}}{\beta ^{2}}\right) ,\qquad \left( \delta x\right)
^{2}=\left( n+\frac{1}{2}\right) \beta ^{-2}  \label{Varpx}
\end{equation}%
in terms of solutions of the Ermakov-type system. In particular,%
\begin{equation}
\left( \delta p\right) ^{2}\left( \delta x\right) ^{2}=\left( n+\frac{1}{2}%
\right) ^{2}\left( 1+\frac{4\alpha ^{2}}{\beta ^{4}}\right) \geq \frac{1}{4}
\label{HeisenbergUR}
\end{equation}%
as required by the fundamental Heisenberg uncertainty relation \cite%
{Bialynicka-Birula87}, \cite{Cor-Sot:Sua:SusInv}, \cite{HeisenbergQM}. The
minimum-uncertainty squeezed states occur for $n=0$ if $\alpha (
t_{\min }) =0.$

\section{Electromagnetic-Field Quantization in Varying Media}

In the macroscopic approach,
one can present the (noncommuting) vector field operators of the electric
displacement $\boldsymbol{D}(\boldsymbol{r},t)$ and the magnetic induction $%
\boldsymbol{B}(\boldsymbol{r},t),$ which fully describe the properties of
the quantized electromagnetic radiation inside a cavity filled with linear
nonstationary dielectric material (with factorized electric
permittivity and magnetic permeability tensors \cite{Dod:Klim:Nik93}), by the
eigenfunction expansions (cf.\ \cite{Bialynicka-Birula87}, \cite{DutraQED},
\cite{HeisenbergQM}, \cite{JaynesCummings63}, \cite{Pedrosa11Conf}, \cite%
{Schleich01})%
\begin{align}
\boldsymbol{D}(\boldsymbol{r},t)& =\sum_{\upsilon }\varpi _{\upsilon }(
t) p_{\upsilon }(t)\boldsymbol{D}_{\upsilon }(\boldsymbol{r}),  \notag
\\
\boldsymbol{B}(\boldsymbol{r},t)& =\sum_{\upsilon }\omega _{\upsilon }(
t) q_{\upsilon }(t)\boldsymbol{B}_{\upsilon }(\boldsymbol{r}),
\label{CavityExpansionsEH}
\end{align}%
in the Heisenberg picture when the time evolution is introduced through the
equations (\ref{pqQED})--(\ref{aacross(t)QED}), which provides a more direct
analogy between quantum and classical physics \cite{Bialynicka-Birula87},
\cite{HarocheRaimond06}. For
a discussion of properties of the stationary orthonormal eigenfunctions $%
\boldsymbol{D}_{\upsilon }(\boldsymbol{r})$ and $\boldsymbol{B}_{\upsilon }(%
\boldsymbol{r})$ defined by the geometry of the cavity and given boundary
conditions, see \cite{BeiLiu11}, \cite{Bialynicka-Birula87}, \cite{Choi04},
\cite{Choi10}, \cite{DodDod05}, \cite{Dod:Klim:Nik93}, \cite{DutraQED},
\cite{Glauber91},
\cite{Knolletal87}, \cite{LobMos91a}, \cite{Manetal10}, \cite{Pedrosa11},
\cite{Pedrosa11Conf}, and \cite{ResEb09}.

In view of the phenomenological Maxwell equations, the single
electromagnetic radiation mode $\upsilon $ in a cavity resonator is
analogous to a parametric (driven) harmonic oscillator \cite{Choi10}, \cite{Dod:Klim:Nik93},
\cite{Fey:Hib}, \cite{Manetal10}, \cite{Pedrosa11}, \cite{Pedrosa11Conf}. This
analogy between classical mechanics and electrodynamics allows one to determine
functions $\omega _{\upsilon }( t) $
and $\varpi _{\upsilon }( t) $ from the electric permittivity,
magnetic permeability, and conductivity of the (slowly-)varying medium in
connection with the quadratic Hamiltonian (\ref{Hamiltonian})
(see Appendix~A for more details).
After
quantization of the field Hamiltonian, the time-dependent operators $%
p_{\upsilon }(t)$ and $q_{\upsilon }(t)$ are determined by
Theorem~1, and the corresponding Fock states are constructed in Lemma~2.
(The average fields obey the classical Maxwell equations \cite%
{Bialynicka-Birula87}.) Methods of stochastic calculus \cite{Oksendal00} can
be used in the case of randomly varying media.

In Schr\"{o}dinger's picture, for the diagonal matrix elements of the field
oscillators, we get%
\begin{align}
\left\langle \boldsymbol{D}(\boldsymbol{r},t)\right\rangle & =\boldsymbol{D}%
_{\upsilon }(\boldsymbol{r})\varpi _{\upsilon }( t) \left\langle
\psi _{n}(t),p\psi _{n}(t)\right\rangle ,  \notag \\
\left\langle \boldsymbol{B}(\boldsymbol{r},t)\right\rangle & =\boldsymbol{B}%
_{\upsilon }(\boldsymbol{r})\omega _{\upsilon }( t) \left\langle
\psi _{n}(t),x\psi _{n}(t)\right\rangle,  \label{EHradiationOperators}
\end{align}%
with (\ref{x2aQED})--(\ref{p2aQED}) for a single mode $\upsilon $, and the
corresponding variances can be obtained with the help of (\ref{Varpx}). In
the autonomous case (see \cite{KrySusVegaMinimum}), the variances are given
(up to a normalization) by equations (A.4)--(A.5) of \cite{LopSusVegaHarm}.

\section{Summary and Applications}

In this letter, the radiation field operators in a cavity with varying
dielectric medium are constructed in terms of explicit solutions of
Heisenberg's equations of motion. The phenomenological quadratic Hamiltonian
under consideration corresponds to the most general (single mode)
one-dimensional mathematical model of quantization in an abstract Hilbert
space. Nonstandard solutions of these equations are obtained
(in terms of quadratures),
and
applications to the radiation field quantization, including randomly varying media,
are briefly discussed.

For most applications in (nonlinear) optics, the electromagnetic field can be
treated classically. But, when quantum limits are approached and one is interested in the photon statistics
of the field, a quantum description is required (see \cite{Drummomd90},
\cite{Glauber91}, \cite{Hillery09}, \cite{Hillery84} and the references therein).
The explicit form of the Bogoliubov transformation (\ref{aacross(t)QED}) (for nonautonomous quadratic systems)
is one of the starting points in this approach \cite{Bialynicka-Birula87}, \cite{Dod:Klim:Nik93},
\cite{Pedrosa11}, \cite{Yuen76}.
The interaction of multiple modes, on the one hand,
and microscopic (lossy) medium models, on the other hand,
are also under consideration \cite{HutterBarnet92},
\cite{Khanetal05}, \cite{Knolletal87}, \cite{Lenac03}, \cite{LobMos91a}, \cite{LobMos91b}, \cite{Raabel07},
\cite{Schuetzholdetal98}, \cite{SuttorpWubs04}.

The problem of quantization of electromagnetic field in material media
remains important in view of recent trends in the flourishing cavity QED
\cite{Dodonov09}, \cite{DodDod05}, \cite{Khanetal05}, \cite{Naylor12}, \cite{Raabel07} and
for experiments in quantum optics in media \cite{Bialynicka-Birula87}, \cite%
{Glauber91}, \cite{GrunerWelsch96}, \cite{Lenac03}, \cite{ResEb09}, \cite{Vasylyevetal10},
which may result in a better understanding of the interaction of light with matter.
Among
other possible applications of the electromagnetic wave propagation in
time-dependent media are the modulation of microwave power \cite%
{Morgenthaler58}, wave propagation in ionized plasma \cite{Kozaki78}, and
magnetoelastic delay lines \cite{RezeMorg69}, \cite{RezeMorg69Exp} (see also
\cite{Choi10} and the references therein).

\medskip

\noindent \textbf{Acknowledgments.\/} We thank Michael Berry for bringing
Ref.~\cite{BerryDennis} to our attention and Victor Dodonov and Luc Vinet for
valuable
discussions. This research was partially supported by an AFOSR grant
FA9550-11-1-0220 and by the Austrian Science Foundation FWF, grants Z130-N13
and S50-N15, the latter in the framework of the Special Research
Program ``Algorithmic and Enumerative Combinatorics".

\appendix

\section*{Appendix: Factorized Media}

\setcounter{equation}{0}%
\global\def\theequation{\mbox{A.\arabic{equation}}}

For the phenomenological Maxwell equations in linear, dispersive,
time-varying media, namely%
\begin{align}
\func{curl}\boldsymbol{E}&=-\frac{1}{c}\frac{\partial \boldsymbol{B}}{%
\partial t},\quad &\func{div}\boldsymbol{D}&=4\pi \rho,\label{Maxwell} \\
\func{curl}\boldsymbol{H}&=\frac{1}{c}\frac{\partial \boldsymbol{D}}{%
\partial t} + \frac{4\pi}{c}\boldsymbol{j},\quad \quad &\func{div}\boldsymbol{B}&=0,
\end{align}%
\begin{equation}
\boldsymbol{D}=\widetilde{\varepsilon }( \boldsymbol{r},t)
\boldsymbol{E},\qquad \boldsymbol{B}=\widetilde{\mu }( \boldsymbol{r}%
,t) \boldsymbol{H},\qquad \boldsymbol{j}=\widetilde{\sigma }( \boldsymbol{r},t)
\boldsymbol{E},
\end{equation}%
the continuity equation,
\begin{equation}
\frac{\partial \rho}{\partial t} + \func{div}\boldsymbol{j} =
\frac{\partial \rho}{\partial t} + \frac{4\pi \widetilde{\sigma }}{\widetilde{\varepsilon }} \rho
+ \boldsymbol{D} \cdot \func{grad} \left( \frac{ \widetilde{\sigma }}{\widetilde{\varepsilon }} \right) = 0 ,
\end{equation}
has the stationary solution $\rho \equiv 0$ under the condition $\func{grad} \left(\widetilde{\sigma }/\widetilde{\varepsilon } \right)=0.$

With the help of the vector $\boldsymbol{A}$ and scalar $\varphi$ potentials,%
\begin{equation}
\boldsymbol{B}=\func{curl}\boldsymbol{A},\qquad \boldsymbol{E}=-\frac{1}{c}%
\frac{\partial \boldsymbol{A}}{\partial t}+ \func{grad}\varphi,
\end{equation}%
the Maxwell equations can be reduced to
the gauge condition
\begin{equation}
\frac{1}{c}\func{div}\left( \widetilde{\varepsilon }\frac{\partial \boldsymbol{A}}{%
\partial t}\right) =\func{div}\left( \widetilde{\varepsilon } \func{grad} \varphi \right)
\end{equation}%
and the single second-order generalized wave equation
\begin{equation}
\func{curl}\left( {\widetilde{\mu }}^{-1} \func{curl}\boldsymbol{A}\right)
+\frac{1}{c^{2}}\frac{\partial }{\partial t}\left( \widetilde{\varepsilon }%
\frac{\partial \boldsymbol{A}}{\partial t}\right)
+ \frac{4\pi \widetilde{\sigma }}{c^2}
\frac{\partial \boldsymbol{A}}{\partial t}  \label{WaveEquation}
 = \frac{1}{c} \frac{\partial }{\partial t}
\left(\widetilde{\varepsilon } \func{grad} \varphi \right)
+ \frac{4\pi \widetilde{\sigma }}{c}\func{grad} \varphi  .
\end{equation}%
Here, we consider
the factorized (real-valued) dielectric
permittivity,
the  magnetic permeability, and
the  conductivity (tensors)%
\begin{equation}
\widetilde{\varepsilon }( \boldsymbol{r},t) =\xi ( t)
\overline{\varepsilon }( \boldsymbol{r}) ,\qquad \widetilde{\mu }%
( \boldsymbol{r},t) =\eta ( t) \overline{\mu }(
\boldsymbol{r}), \qquad \widetilde{\sigma }%
( \boldsymbol{r},t) =\chi ( t) \overline{\sigma }(
\boldsymbol{r}) \label{DKT}
\end{equation}%
(the case $\widetilde{\sigma }\equiv 0$ was discussed in \cite{Dod:Klim:Nik93}).
Under the imposed condition of $\func{grad} \left(\widetilde{\sigma }/\widetilde{\varepsilon } \right)=0,$
one can choose $4\pi\overline{\sigma }=\overline{\varepsilon }$ without loss of generality.

The solution of the classical problem for a given single mode,
$\upsilon $ say,
has the form
$$\boldsymbol{A}( \boldsymbol{r},t) =\boldsymbol{u}%
( \boldsymbol{r}) q( t), \quad \quad
\varphi(\boldsymbol{r},t) =(k/c)\dfrac{dq}{dt} \phi(\boldsymbol{r})$$
($k$ is a constant), and%
\begin{equation}
\boldsymbol{B}=q\func{curl}\boldsymbol{U},\qquad \boldsymbol{D}=-\frac{\xi }{%
c}\frac{dq}{dt}\ \overline{\varepsilon }\boldsymbol{U},
\end{equation}%
provided that%
\begin{alignat}2
\notag
\func{curl}\left( \frac{1}{\overline{\mu }}\func{curl}\boldsymbol{U}%
\right) &=\upsilon ^{2}\, \overline{\varepsilon }\boldsymbol{U},\qquad &%
\boldsymbol{U} &=\boldsymbol{u}-k\func{grad}\phi, \\
\frac{d^{2}q}{dt^{2}}+\frac{\xi ^{\prime }+\chi}{\xi }\frac{dq}{dt}+\frac{%
c^{2}\upsilon ^{2}}{\xi \eta }q&=0, \ &\upsilon &=\text{constant} ,
\label{EigenValueProblem}
\end{alignat}%
and certain required boundary conditions are satisfied on the boundary of the cavity
(see \cite{Dod:Klim:Nik93}, \cite{DutraQED}, \cite{Glauber91}, \cite{Knolletal87}, \cite{Schleich01} for more details).

As a result, we can choose $c=d=f=g=0$ and
\begin{equation}
a=\frac{1}{2\xi} e^{-\int(\chi/\xi)\, dt}, \qquad b=\frac{c^{2}\upsilon ^{2}}{2\eta}
e^{\int(\chi/\xi)\, dt}
\end{equation}
in Theorem~1 for the quantization of the mode of the electromagnetic field under consideration.
(See also \cite{BeiLiu11}, \cite{Choi10}, \cite{Manetal10}, \cite{Pedrosa11}, and \cite{Pedrosa11Conf}.)

\end{document}